\documentclass[a4paper,twoside,11pt]{article}
\usepackage[margin=2.54cm]{geometry}
\usepackage[affil-it]{authblk}
\usepackage{amsmath,amsfonts,amsthm,amssymb}
\usepackage{graphics,graphicx,color}
\usepackage{enumerate,paralist}
\usepackage[labelsep=period,labelfont=bf,size=small,compatibility=false]{caption}
\usepackage{subfig,xspace}
\usepackage{todonotes}

\graphicspath{{figures/}}

\title{The Book Thickness of 1-Planar Graphs is Constant}

\author{Michael A. Bekos%
\thanks{Electronic address: \texttt{bekos@informatik.uni-tuebingen.de}}}

\author{Till Bruckdorfer%
\thanks{Electronic address: \texttt{bruckdor@informatik.uni-tuebingen.de}}}

\author{Michael~Kaufmann%
\thanks{Electronic address: \texttt{mk@informatik.uni-tuebingen.de}}}
\affil{Wilhelm-Schickhard-Institut f\"ur Informatik, Universit\"at T\"ubingen, Germany}

\author{Chrysanthi~N.~Raftopoulou%
\thanks{Electronic address: \texttt{crisraft@mail.ntua.gr}}}
\affil{School of Applied Mathematical \& Physical Sciences, NTUA, Greece.}

\date{}

\makeatletter
\def\@maketitle{%
  \newpage
  \null
  \vskip 2em%
  \begin{center}%
  \let \footnote \thanks
    {\Large\bfseries \@title \par}%
    \vskip 1.5em%
    {\normalsize
      \lineskip .5em%
      \begin{tabular}[t]{c}%
        \@author
      \end{tabular}\par}%
    \vskip 1em%
    {\normalsize \@date}%
  \end{center}%
  \par
  \vskip 1.5em}
\makeatother

\newtheorem{lemma}{Lemma}
\newtheorem{theorem}{Theorem}

\begin{document}

\maketitle

\begin{abstract}
In a book embedding, the vertices of a graph are placed on the
``spine'' of a book and the edges are assigned to ``pages'', so that
edges on the same page do not cross. In this paper, we prove that
every $1$-planar graph (that is, a graph that can be drawn on the
plane such that no edge is crossed more than once) admits an
embedding in a book with constant number of pages. To the best of
our knowledge, the best non-trivial previous upper-bound is
$O(\sqrt{n})$, where $n$ is the number of vertices of the graph.
\end{abstract}

\section{Introduction}
\label{sec:intro}

A \emph{book embedding} is a special type of a graph embedding, in
which
\begin{inparaenum}[(i)] \item the vertices of the graph are
restricted to a line along the \emph{spine} of a book, and, \item
the edges on the \emph{pages} of the book in such a way that edges
residing on the same page do not cross.
\end{inparaenum}
The minimum number of pages required to construct such an embedding
is known as \emph{book thickness} or \emph{page number} of a graph.
An obvious upper bound on the page number of an $n$-vertex graph is
$\lceil n/2 \rceil$, which is tight for complete graphs~\cite{BK79}.
Book embeddings have a long history of research dating back to early
seventies~\cite{Tay73}. Therefore, there is a rich body
of literature (see, e.g.,~\cite{Bil92} and \cite{Ove07}).

For the class of planar graphs, a central result is due to
Yannakakis~\cite{Yan89}, who in the late eighties proved that planar
graphs have book thickness at most four. It remains, however,
unanswered whether the known bound of four is tight.
Heath~\cite{Hea84}, for example, proves that all planar $3$-trees
are $3$-page book embeddable. For more restricted subclasses of planar
graphs, Bernhart and Kainen~\cite{BK79} show that the graphs with
book thickness one are the outerplanar graphs, while the class of
two-page embeddable graphs coincides with the class of
subhamiltonian graphs (recall that \emph{subhamiltonian} is a graph
that is a subgraph of a planar Hamiltonian graph). Testing whether a
graph is subhamiltonian is NP-complete~\cite{Whi31}. However,
several graph classes are known to be subhamiltonian (and therefore
two-page book embeddable), e.g., $4$-connected planar
graphs~\cite{nc-c10hc-PG08}, planar graphs without separating
triangles~\cite{ko-etw-AML07}, Halin graphs~\cite{cnp-hgtsp-MP83},
planar graphs with maximum degree $3$ or
$4$~\cite{h-afegib-UNC85,bgr-tpbe4-STACS14}.

In this paper, we go a step beyond planar graphs. In particular, we
consider $1$-planar graphs and prove that their book thickness is
constant. Recall that a graph is \emph{$1$-planar}, if it admits a
drawing in which each edge is crossed at most once. To the best of
our knowledge, the only (non-trivial) upper bound on the book
thickness of $1$-planar graphs on $n$ vertices is $O(\sqrt{n})$.
This is due to two known results: First, graphs with $m$ edges have
book thickness $O(\sqrt{m})$~\cite{Malitz1994}. Second, $1$-planar
graphs with $n$ vertices has at most $4n-8$ edges, which is a tight
bound~\cite{BSW84,FM07,PT97}. Minor-closed graphs (e.g., graphs of
constant treewidth~\cite{DW07} or genus~\cite{Mal94}) have constant
book thickness~\cite{NM12}. Unfortunately, however, $1$-planar
graphs are not closed under minors~\cite{NM12}. 

In the remainder of this paper, we will assume that a simple
$1$-planar drawing $\Gamma(G)$ of the input $1$-planar graph $G$ is
also specified as part of the input of the problem. This is due
to a result of Grigoriev and Bodlaender~\cite{GB07}, and,
independently of Kohrzik and Mohar~\cite{KM13}, who proved that the
problem of determining whether a graph is $1$-planar is NP-hard
(note that the problem remains NP-hard, even if the deletion of a
single edge makes the graph planar~\cite{CM12}). In addition, we
assume biconectivity, as it is known that the page number of a graph
equals to the page number of its biconnected components~\cite{BK79}.

\section{Definitions and Yannakakis Algorithm}
\label{sec:preliminaries}

Let $G$ be a simple topological graph, that is, undirected and drawn
in the plane. We denote by $\Gamma(G)$ the drawing of $G$.
Unless otherwise specified, we consider \emph{simple} drawings, that
is, no edge crosses itself, no two edges meet tangentially and no
two edges cross more than once. A drawing uniquely defines the
cyclic order of the edges incident to each vertex and, therefore,
specifies a \emph{combinatorial embedding}. A $1$-planar topological
graph is called \emph{planar-maximal} or simply \emph{maximal}, if
the addition of a non-crossed edge is not possible. The following
lemma, proven in many earlier papers, shows that two crossing edges
induce a $K_4$, as the missing edges can be added without
introducing new crossings; see, e.g.,~\cite{ABK13}.

\begin{lemma}
In a maximal $1$-planar topological graph, the endpoints of two
crossing edges are pairwise adjacent. 
\label{lem:lem_crossing_quadrangle}
\end{lemma}

The base of our approach is the simple version of Yannakakis
algorithm, which embeds any (internally-triangulated) plane graph in
a book of five pages~\cite{Yan89}; not four. In the following, we
outline this algorithm. However, we assume basic familiarity. The
algorithm is based on a ``peeling'' into \emph{levels} approach. In
particular,
\begin{inparaenum}[(i)]
\item vertices on the outerface are at level zero;
\item vertices that are on the outerface of the graph induced by
deleting all vertices of levels $\leq i-1$ are at level $i$;
\item edges between vertices of the same (different, resp.)
level are called \emph{level} (\emph{binding}, resp.) edges; see
Fig.~\ref{fig:yannakakis}.
\end{inparaenum}
In a high-level description, the algorithm first embeds level zero
followed by level one and the binding edges between levels zero and
one. The remaining graph, that is in the interior of all level-one
cycles, is embedded recursively.

\subsection{The two-level case}
To achieve a total of five pages, first it is proven that a graph
$G=(V,E)$, consisting only of two levels, say $L_0$ and $L_1$, is
three page embeddable (it is also assumed that $L_0$ has no chords).
The vertices, say $u_1,u_2,\ldots,u_k$, of $L_0$ are called
\emph{outer} and appear in this order along the clockwise traversal
of the outerface of $G$. The remaining vertices are called
\emph{inner} (and obviously belong to $L_1$).  The graph induced by
all outer vertices is biconnected. The biconnected components (or
\emph{blocks}), say $B_1,B_2,\ldots,B_m$, of the graph induced by
the inner vertices form a tree (in the absence of chords in $L_0$).
It is assumed that the block-tree is rooted at the block, w.l.o.g.
at block $B_1$, that contains the so-called \emph{first inner
vertex}, which is uniquely defined as the third vertex of the
bounded face containing the outer vertices $u_1$ and $u_k$. Given a
block $B_i$, an outer vertex is said to be \emph{adjacent} to $B_i$
if it is adjacent to a vertex of it; the set of outer vertices
adjacent to $B_i$ is denoted by $N(B_i)$ , $i=1,2,\ldots,m$.
Furthermore, a vertex $w$ is said to \emph{see} an edge $(x,y)$, if
$w$ is adjacent to $x$ and $y$ and the triangle $x,y,w$ is a face.
An outer vertex sees a block if it sees an edge of it.

\begin{figure}[t!]
    \centering
    \begin{minipage}[b]{.40\textwidth}
        \centering
        \subfloat[\label{fig:yannakakis-input}{An internally-triangulated graph.}]
        {\includegraphics[width=\textwidth,page=1]{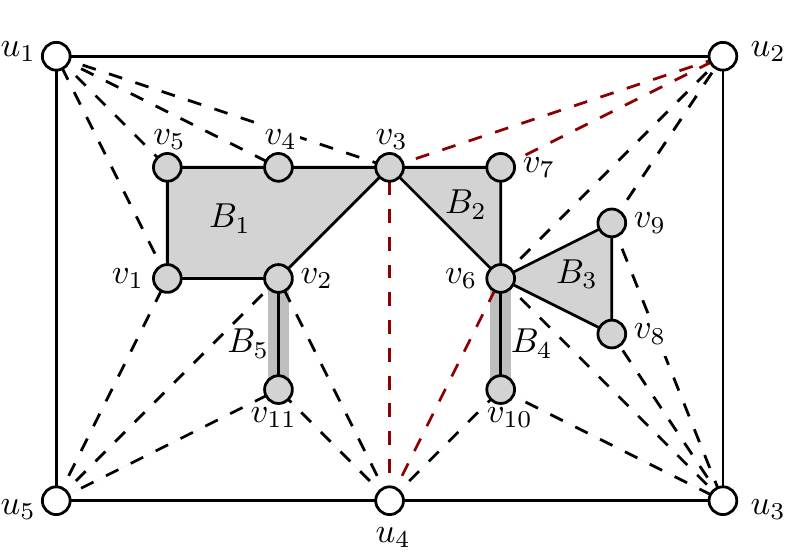}}
    \end{minipage}
    \hfill
    \begin{minipage}[b]{.58\textwidth}
        \centering
        \subfloat[\label{fig:yannakakis-output}{A book embedding in three pages; taken from~\cite{Yan89}.}] 
        {\includegraphics[width=\textwidth,page=2]{yannakakis}}
    \end{minipage}
    \caption{%
    (a)~Outer (inner) vertices are colored white (gray).
    Level (binding) edges are solid (dashed).
    Blocks are highlighted in gray.
    The first inner vertex is $v_1$.
    So, the root of the block tree is $B_1$.
    $N(B_3)=\{u_2,u_3\}$.
    Vertex $u_2$ sees $(v_3,v_7)$ and so sees $B_2$.
    The leaders of $B_1$, $B_2$, $B_3$, $B_4$ and $B_5$ are
    $v_1$, $v_3$, $v_6$, $v_6$ and $v_2$, resp.
    The dominators of $B_1$, $B_2$, $B_3$, $B_4$ and $B_5$ are 
    $u_1$, $u_2$, $u_2$, $u_3$ and $u_4$, resp.; 
    The red edges indicate that $u_f(B_2)=u_2$ and $u_l(B_2)=u_4$.
    Hence, $P[u_f(B_2) \rightarrow u_l(B_2)] = u_2 \rightarrow u_3 
    \rightarrow u_4$.
    (b)~Linear order and assignment of edges to pages.}
    \label{fig:yannakakis}
\end{figure}

The \emph{leader} of a block $B_i$ is the first vertex of block
$B_i$ that is encountered in any path from the first inner vertex to
block $B_i$ and is denoted by $\ell(B_i)$. An inner vertex that
belongs to only one block is \emph{assigned} to that block. One that
belongs to more than one blocks is assigned to the ``highest block''
in the block-tree that contains it. Given an inner vertex $v \in
L_1$, we denote by $B(v)$ the block that $v$ is assigned to. The
\emph{dominator} of a block $B_i$ is the first outer vertex that is
adjacent to a vertex of block $B_i$ and is denoted by $dom(B_i)$,
$i=1,2,\ldots,m$.

Let $B$ be a block of level $L_1$ and assume that
$v_0,v_1,\dots,v_t$ are the vertices of $B$ as they appear in a
counterclockwise traversal of the boundary of $B$ starting
$v_0=\ell(B)$. Denote by $u_f(B)$ and $u_l(B)$ the smallest- and
largest-indexed vertices of level $L_0$ that see edges $(v_0,v_k)$
and $(v_0,v_1)$, respectively.
Equivalently, $u_f(B)$ and $u_l(B)$ are defined as the smallest- and
largest-indexed vertices of $N(B)$. Note that $u_f(B)=dom(B)$.
The path on level $L_0$ from $u_f(B)$ to $u_l(B)$ in clockwise
direction along $L_0$ is denoted by $P[u_f(B) \rightarrow u_l(B)]$.

\subsection{Linear order}
The \emph{linear order} of the vertices along the spine is computed
as follows. First, the outer vertices are embedded in the order
$u_1,u_2,\ldots,u_k$. For $j=1,2,\ldots,k$, the blocks dominated by
the outer vertex $u_j$ are embedded right next to $u_j$ one after
the other in the top to down order of the block-tree. The vertices
that belong to block $B_i$ are ordered along the spine in the order
that appear in the counterclockwise traversal of the boundary of
$B_i$ starting from $\ell(B_i)$, $i=1,2,\ldots,m$ (which is already
placed).

\subsection{Edge-to-page assignment}
The edges are assigned to pages as follows. All level edges of $L_0$
are assigned to the first page. Level-one edges are assigned either to
the second or to the third page, based on whether they belong to a
block that is in an odd or even distance from the root of the block
tree, respectively. Binding edges are further classified as
\emph{forward} or \emph{back}. A binding edge is forward if the
inner vertex precedes the outer vertex; otherwise it is back (recall
that a binding edge connects an outer and an inner vertex). All back
edges are assigned to the first page. A forward edge incident to a
block $B_i$ is assigned to the second page, if $B_i$ is on the third
page; otherwise to the third page, $i=1,2,\ldots,m$.

\subsection{The multi-level case}
Possible chords in level $L_0$ are assigned to the first page. Note,
however, that in the presence of such chords, the blocks of level
$L_1$ form a forest in general (i.e., not a single tree). Therefore,
each block-tree of the underlying forest must be embedded according
to the rules described above. Graphs with more than two layers are
embedded by ``recycling'' the remaining available pages. More
precisely, consider a block $B$ of level $i-1$ and let $B'$ be a
block of level $i$ that is in the interior of $B$ in the peeling
order. Let $\{p_1,\ldots,p_5\}$ be a permutation of $\{1,\ldots,5\}$
and assume w.l.o.g. that the boundary of block $B$ is assigned to
page $p_1$, while the boundary of all blocks in its interior
(including $B'$) are assigned to pages $p_2$ and $p_3$. Then, the
boundary of all blocks of level $i+1$ that are in the interior of
$B'$ in the peeling order will be assigned to pages $p_4$ and $p_5$.
The correctness of this strategy follows from the fact that all
blocks of level $i+1$ that are in the interior of a certain block of
level $i$ are always between two consecutive vertices of level
$i-1$. This directly implies that blocks that are by at least two
levels apart in the peeling order are in a sense independent, which
allows pages that have already been used by some previous levels to
be reused by blocks of next levels provided that they are not
consecutive. 
In the following we present properties that we use in the remainder
of the paper.

\begin{lemma}[Yannakakis~\cite{Yan89}]
Let $G$ be a graph consisting of two levels $L_0$ and $L_1$. Let $B$
be a block of level $L_1$ and let $v_0,\dots,v_t$ be the vertices of
$B$ in a counterclockwise order along the boundary of $B$ starting
from $v_0=\ell(B)$. Then:
\begin{enumerate}[\emph{(}i\emph{)}]
\item \label{lem:cons_on_spine} Vertices $v_1,\dots,v_t$ are
consecutive along the spine.
\item \label{lem:diff_first_last} $u_f(B)\neq u_l(B)$.
\item \label{lem:lem_crossing_basics} If $u_i=u_f(B)$ and
$u_j=u_l(B)$ for some $i<j$, then vertices $v_1, \dots ,v_t,
u_{i+1}, \dots, u_j$ appear in this order from left to right along
the spine.
\item \label{lem:block_order} Let $G[B]$ be the subgraph of $G$ in
the interior of cycle $P[u_f(B) \rightarrow u_l(B)] \rightarrow
\ell(B) \rightarrow  u_f(B)$. Then, a block $B'\in G[B]$ if and only
if $B$ is an ancestor of $B'$, that is, $B'$ belongs to the
block-subtree rooted at $B$.
\end{enumerate}
\label{lem:lem_crossing_ordering_1}
\end{lemma}

\begin{lemma}[Yannakakis~\cite{Yan89}]
Let $G$ be a graph consisting of two levels $L_0$ and $L_1$ and
assume that $(u_i,u_j)$, $i<j$, is a chord of $L_0$. Denote by $H$
the subgraph of $G$ in the interior of the cycle $P[u_i \rightarrow
u_j] \rightarrow u_i$. Then:
\begin{enumerate}[\emph{(}i\emph{)}]
\item Vertices $u_i$ and $u_j$ form a separation pair in $G$.
\item All vertices of $H$ lie between $u_i$ and $u_j$ along the spine.
\item If there is a vertex between $u_i$ and $u_j$ that does not
belong to $H$, then this vertex belongs to a block $B$ dominated by
$u_i$. In addition, all vertices of $H$, except for $u_i$ are to the
right of $B$ along the spine.
\end{enumerate}
\label{lem:lem_crossing_chord_separating}
\end{lemma}

\section{An upper bound on the book thickness of 1-Planar Graphs}
\label{sec:1planar}

In this section, we extend the algorithm of Yannakakis~\cite{Yan89}
to $1$-planar graphs based on the ``peeling'' approach described in
Section~\ref{sec:preliminaries}. Let $G=(V,E)$ be a $1$-planar graph
and $\Gamma(G)$ be a $1$-planar drawing of $G$. Our approach, in
high level description, is as follows. Initially, we consider the
case where $\Gamma(G)$ contains no crossings incident to its
unbounded face. We augment $G$ to internally maximal $1$-planar. To
do so, we momentarily replace each crossing in $\Gamma(G)$ with a
so-called \emph{crossing vertex}. The implied planarized graph is
then triangulated (only in its interior), so that no new edge is
incident to a crossing vertex. The latter restriction, however, may
lead to a non-simple graph (containing multiedges), as we will see
in Section~\ref{subsec:nonmaximal}. On the other hand, the interior
of all $K_4$s implied by Lemma~\ref{lem:lem_crossing_quadrangle} are
free of vertices and edges. To simplify the presentation, we
initially assume that the planarized graph is simple. So, $G$ can be
augmented to a simple internally maximal $1$-planar graph with no
crossings incident to its unbounded face. Following an approach
slightly different from the one of Yannakakis~\cite{Yan89}, we show
how one can define the levels for such a graph. We prove that if
there are only two levels, then such a graph fits in $16$ pages.
When the number of levels is greater than two, we prove that $39$
pages suffice. Finally, we show how one can cope with the cases of
multiedges and crossings on the unbounded face of $\Gamma(G)$.

Assume now that $G$ is simple and internally maximal $1$-planar with
no crossings incident to the unbounded face of $\Gamma(G)$. Its
vertices are assigned to levels as follows:
\begin{inparaenum}[(i)]
\item vertices on the outerface of $G$ are at level zero;
\item vertices that are at distance $i$ from the level zero
vertices are at level $i$.
\end{inparaenum}
Similarly to Yannakakis naming scheme, edges that connect vertices
of the same (different, resp.) level are called \emph{level}
(\emph{binding}, resp.) edges.

Since we have assumed that $G$ is internally maximal $1$-planar and
that there are no crossings incident to its unbounded face, by
Lemma~\ref{lem:lem_crossing_quadrangle} it follows that the
endpoints of every crossing pair of edges are pairwise adjacent. So,
if we remove one edge from each pair of crossing edges, then the
result is an internally-triangulated plane graph (which we call
\emph{underlying planar structure}). We apply the following simple
rule. For a pair of binding crossing edges or for a pair of level
crossing edges, we choose arbitrarily one to remove\footnote{We will
shortly adjust this choice for two special cases. However, since we
do not seek to enter into further details at this point, we assume
for now that the choice is, in general, arbitrary.}. However, for a
pair of crossing edges consisting of a binding edge and a level
edge, we always choose to remove the level edge. The main benefit of
the aforementioned approach is that, the underlying planar structure
allow us to define the blocks as well as the leaders and the
dominators of the blocks in the exact same way as Yannakakis does.
In addition, it is not difficult to observe that if all removed
edges are plugged back to the graph, then a binding edge cannot
cross a block, since such a crossing would involve an edge on the
boundary of the block, which by definition is a level edge (and
therefore not present when the blocks are computed).

\subsection{The Two-Level Case}
\label{sec:twolevels}

In this section, we consider the (intuitively easier) case, where
the given $1$-planar graph $G$ consists of two levels $L_0$ and
$L_1$. We also assume that there is no level edge of $L_1$ which by
the combinatorial embedding is strictly in the interior of a block
of $L_1$. In addition, $G$ is simple internally maximal $1$-planar
and has no crossings on its unbounded face. We denote by $G_P$ the
underlying planar structure and proceed to obtain a $3$-page book
embedding of $G_P$ using Yannakakis algorithm~\cite{Yan89}
(see Section~\ref{sec:preliminaries}). We argue that we can
embed the removed edges in the linear order implied by the book
embedding of $G_P$ using $11$ more pages.

Before we proceed with the description of our approach, we introduce
an important notion useful in ``eliminating'' possible crossing
situations. We say that two edges $e_1$ and $e_2$ of $G$
\emph{form a strong pair} if 
\begin{inparaenum}[(i)] 
\item they are both assigned to the same page, say $p$, and,
\item if an edge $e$, that is assigned also to page $p$, crosses
$e_i$, then it also crosses $e_j$, where $i\neq j\in
\left\{1,2\right\}$.
\end{inparaenum}
Suppose that $e_1\notin E[G_P]$ and $e_2\in E[G_P]$ form a strong
pair of edges. If $e_3\in E[G_P]$, then $e_3$ can cross neither
$e_1$ nor $e_2$ (due to Yannakakis' algorithm). On the other hand,
if $e_3\notin E[G_P]$ and forms a strong pair  with another edge
$e_4\in E[G_P]$, then again $e_3$ can cross neither $e_1$ nor $e_2$,
as otherwise $e_4$ would also be involved in a crossing with $e_1$
or $e_2$; contradicting the correctness of Yannakakis' algorithm as
$e_4\in E[G_P]$.

In the following, we describe six types of crossings that may occur
when the removed edges are plugged back to $G$; see
Fig.~\ref{fig:1-planar}. Level edges of $L_0$ that do not belong to
the planar structure $G_P$ are called \emph{outer crossing chords}.
Such chords may be involved in crossings with
\begin{inparaenum}[(i)] 
\item other chords of $L_0$ that belong to $G_P$, or,
\item binding edges (between levels $L_0$ and $L_1$), or, 
\item degenerated blocks (so-called \emph{block-bridges}) of level $L_1$ that are simple edges.
\end{inparaenum}

\begin{figure}[t!]
    \centering
    \includegraphics[width=.4\textwidth,page=1]{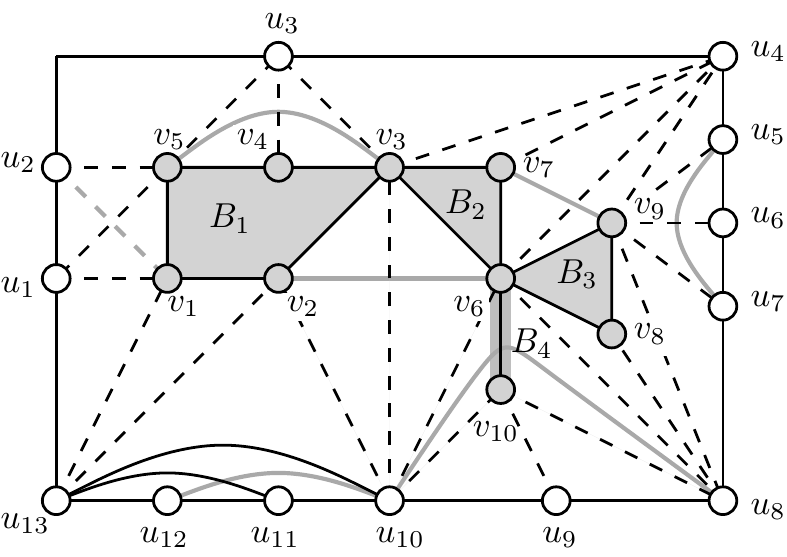}
    \caption{%
    Level (binding) edges are solid (dashed).
    The planar structure is colored black.
    Gray colored edges do not belong to the planar structure.
    Edges $(u_5,u_7)$, $(u_8,u_9)$ and $(u_{10},u_{12})$ are outer
    crossing chords that cross a binding edge, a bridge-block and a
    chord of the planar structure, resp.
    Edges $(v_3,v_5)$, $(v_7,v_9)$ and $(v_2,v_6)$ are $2$-hops crossing binding edges.
    Edge $(u_1,v_5)$ is forward binding edge.}
    \label{fig:1-planar}
\end{figure}

Level edges of $L_1$ that do not belong to the planar structure are
called \emph{inner crossing chords} or simply
\emph{$2$-hops}\footnote{The term yields from the observation that
along the boundary of the block-tree a $2$-hop can ``bypass'' only
one vertex because of maximal 1-planarity.}. We claim that $2$-hops
do not cross with each other. For a proof by contradiction, assume
that $e=(u,v)$ and $e'=(u',v')$ are two $2$-hops that cross. Assume
w.l.o.g. that $u$, $u'$, $v$ and $v'$ appear in this order along the
boundary of the block tree of the underlying planar structure $G_P$.
Since $G$ is maximal $1$-planar, by
Lemma~\ref{lem:lem_crossing_quadrangle} it follows that $(u,v')$
belongs to the planar structure $G_P$. On the other hand, in the
presence of this edge both vertices $u'$ and $v$ are not anymore at
the boundary of the block tree of level $L_1$ of $G_P$, which is a
contradiction as $e$ and $e'$ are both level edges of $L_1$. Hence,
$2$-hops are involved in crossings only with binding edges. Since
level edges of different levels cannot cross, the only type of
crossings that we have not reported so far are those between binding
edges.

Recall that binding edges are of two types; forward and back.
For a pair of crossing binding edges, say $e=(u_i,v_j)$ and
$e'=(u_{i'},v_{j'})$, where $u_i,u_{i'}\in L_0$ and $v_i,v_{i'}\in
L_1$, we mentioned that we can arbitrarily choose which one is
assigned to the underlying planar structure $G_P$. Here, we adjust
this choice: Edge $e$ is assigned to $G_P$ if and only if vertex
$u_i$ is lower-indexed than $u_{i'}$ in $L_0$, that is, $i < i'$.
As a consequence, edge $e'$ is always forward binding. Therefore,
two back binding edges cannot cross.

Similarly, for a pair of crossing level edges, we mentioned that we
arbitrarily choose, which one to assign to the underlying planar
structure $G_P$. We adjust this choice in the case of two crossing
level edges of level $L_0$ as follows: If a level edge is incident
to the first outer vertex $u_1$ of level $L_0$, then it is
necessarily assigned to $G_P$.

From the above, it follows that for a pair of crossing edges, say $e
\in E[G_P]$ and $e' \notin E[G_P]$, we have the following crossing
situations, each of which is separately treated in the following
lemmas (except for the last one which is more demanding):

\begin{enumerate}[C.1:]
\item $e'$ is an outer crossing chord and $e$ is a chord of $L_0$ that belongs to $G_P$.
\label{case:chord_chord}
\item $e'$ is an outer crossing chord and $e$ is a binding edge.
\label{case:chord_binding}
\item $e'$ is an outer crossing chord and $e$ is a block-bridge of $L_1$.
\label{case:chord_bridge}
\item $e'$ is a forward binding edge and $e$ is a forward binding edge.
\label{case:forward_forward}
\item $e'$ is a forward binding edge and $e$ is a back binding edge.
\label{case:forward_back}
\item $e'$ is a $2$-hop and $e$ is a binding edge.
\label{case:hop_binding}
\end{enumerate}

\noindent Our approach is outlined in the proof of the following
theorem, which summarizes the main result of this section.

\begin{theorem}
Any simple internally maximal $1$-planar graph $G$ with $2$ levels
and no crossings incident to its unbounded face admits a book
embedding on $16$ pages.
\end{theorem}
\begin{proof}
The underlying planar structure can be embedded in three pages.
Case~C.\ref{case:chord_chord} requires one extra page; due to
Lemma~\ref{lem:lem_crossing_inner_chord}.\ref{lem:chord_chord}.
The crossing edges that fall into Cases~C.\ref{case:chord_binding}
and C.\ref{case:chord_bridge} can be accommodated on the same pages
used for the underlying planar structure; see
Lemma~\ref{lem:lem_crossing_inner_chord}.\ref{lem:chord_binding_bridge}.
Case~C.\ref{case:forward_forward} requires two extra pages due to
Lemma~\ref{lem:lem_crossing_binding_back}. 
Case~C.\ref{case:forward_back} requires three extra pages due to
Lemma~\ref{lem:lem_crossing_binding_forward}. Finally, Case
C.\ref{case:hop_binding} requires seven more pages due to
Lemma~\ref{lem:two_hops}. Summing up the above yields $16$ pages in
total.
\end{proof}

We start by investigating the case where $e'$ is an outer crossing
chord of $G$. By Lemma~\ref{lem:lem_crossing_inner_chord_auxiliary}
any two outer crossing chords can be placed on the same page without
crossing each other, while Lemma~\ref{lem:lem_crossing_inner_chord}
describes the placement of outer crossing chords for Cases
C.\ref{case:chord_chord}, C.\ref{case:chord_binding} and
C.\ref{case:chord_bridge}.

\begin{lemma}
Let $u_1,\dots,u_k$ be the vertices of level $L_0$ in clockwise
order along its boundary. Let $c=(u_i,u_j)$ and $c'=(u_{i'},u_{j'})$ be
two chords of $L_0$, such that $i<i'<j<j'$. Then, exactly one of $c$
and $c'$ is an outer crossing chord.
\label{lem:lem_crossing_inner_chord_auxiliary}
\end{lemma}
\begin{proof}
Since $c=(u_i,u_j)$ and $c'=(u_{i'},u_{j'})$ are chords of $L_0$
with $i<i'<j<j'$, $c$ and $c'$ cross in the $1$-planar drawing
$\Gamma(G)$ of $G$. So, one of them would belong to the underlying
planar structure of $G$ and the other one would be an outer crossing
chord.
\end{proof}

Lemma~\ref{lem:lem_crossing_inner_chord_auxiliary} implies that
outer crossing chords can be placed on one page without crossing.
However, as stated in the following lemma, we choose not to do so.
Recall that we use three pages for $G_P$; $p_1$, $p_2$ and $p_3$.
Page $p_1$ is devoted to level edges of $L_0$ and back binding edges
of $G_P$. Pages $p_2$ and $p_3$ are used for level edges of $L_1$
and forward binding edges.

\begin{lemma}[Cases C.\ref{case:chord_chord} - C.\ref{case:chord_bridge}] 
Let $e=(u,v)\in E(G_P)$ and $e'=(u',v')\notin E(G_P)$ be two edges
of $G$ that are involved in a crossing. If $e'$ is outer crossing
chord of $L_0$, then:
\begin{enumerate}[\emph{(}i\emph{)}]
\item \label{lem:chord_chord} If $e$ is a chord of $L_0$, then $e'$
is placed on a universal page denoted by $up_c$
(Case~C.\ref{case:chord_chord}).
\item If $e$ is a binding or a block-bridge of $L_1$, then $e'$ is
assigned to page $p_1$, that is, the page used for level
edges of $L_0$ and back binding edges of $G_P$
(Cases~C.\ref{case:chord_binding} and C.\ref{case:chord_bridge}).
\label{lem:chord_binding_bridge}
\end{enumerate}
\label{lem:lem_crossing_inner_chord}
\end{lemma}
\begin{proof}
\begin{inparaenum}[\emph{(}i\emph{)}]
\item Since $e$ is chord of level $L_0$, $e$ is placed on page $p_1$
(recall that $e\in E(G_P)$), and $e'$ is placed on the universal
page $up_c$. Since $up_c$ contains only outer crossing chords of
$G$, by Lemma~\ref{lem:lem_crossing_inner_chord_auxiliary} they do
not cross with each other.

\item If $e$ is a binding or a block-bridge of level $L_1$, then
$e'$ is assigned to page $p_1$. Suppose that $e'$ is in conflict with
another edge, say $e''$, of page $p_1$. By
Lemma~\ref{lem:lem_crossing_inner_chord_auxiliary}, edge $e''$ is
not an outer crossing chord, that is, $e''$ belongs to the
underlying planar structure $G_P$ of $G$. So, $e''\in E[G_P]$ and
it is either:
\begin{inparaenum}[(a)]
\item a level edge of level $L_0$, or
\item a back binding edge of $G_P$.
\end{inparaenum}
In the first case, the endpoints of $e''$ cannot be consecutive
vertices of level $L_0$, since that would not lead to a crossing
situation. Hence, $e''$ must be a chord of level $L_0$. However, if
$e'$ is involved in such a crossing, then $e'$ is assigned to page
$up_c$; a contradiction. In the second case, $e''$ is back binding
of $G_P$. So, edge $e''$ is nested by a level edge of level $L_0$
and, therefore, if $e'$ crosses $e''$, then $e'$ must also cross
this particular level edge of level $L_0$, which is not possible.
\end{inparaenum}
\end{proof}


\begin{lemma}[Case C.\ref{case:forward_forward}]
All forward binding edges that are involved in crossings with forward
binding edges of the underlying planar structure can be assigned to
$2$ new pages.
\label{lem:lem_crossing_binding_back}
\end{lemma}
\begin{proof}
To prove the lemma, we employ a simple trick. We observe that, for a
pair of crossing forward binding edges, the choice of the edge that
will be assigned to the underlying planar structure affects neither
the decomposition into blocks nor the choice of dominators and
leaders of blocks. Therefore, it does not affect the linear order of
the vertices along the spine. This ensure that two new pages
suffice.
\end{proof}

We proceed with Case C.\ref{case:forward_back}, where the back
binding edge $e=(u,v) \in E[G_P]$ crosses the forward binding
$e'=(u',v') \notin E[G_P]$. Let $P$ be the block containing $(v,v')$
and let $v_0,v_1,\dots,v_t$ be the vertices of $P$ as they appear in
the counterclockwise order around $P$ starting from $v_0=\ell(P)$.
Since $e$ is back binding, it follows that $u=u_f(P)$. By definition
of $u_f(P)$, $u$ sees edges $(v_i,v_{i+1})$, \dots, $(v_{t-1},v_t)$,
$(v_t,v_0)$ of $P$, for some $1\leq i\leq t$. Hence, edges
$(u,v_0)$, $(u,v_t)$, \dots, $(u,v_i)$ exist and are back binding
edges. This implies that either $v=v_i$ and $v'=v_{(i+1)modt}$ or
$v=v_0$ and $v'=v_t$.  In the latter case and assuming that $u'$ is
to the right of $u$ on the spine, $P$ is a root-block. In both
cases, $(u',v)$ is forward.

\begin{lemma}[Case C.\ref{case:forward_back}]
Let $e=(u,v)$ be a back binding edge and $e'=(u',v')$ a forward
binding edge of $G$ that cross. Let also $v_0,v_1,\dots,v_t$ be the
vertices of block $P$ as they appear in the counterclockwise order
around $P$ starting from $v_0=\ell(P)$, where $P$ is the block
containing $(v,v')$. Finally, let $i$ be the minimum s.t. vertex
$u=u_f(P)$ sees edges $(v_i,v_{i+1})$, \dots, $(v_{t-1},v_t)$,
$(v_t,v_0)$. Then we use three new pages $p'_1$, $p'_2$ and $p'_3$
as follows:
\begin{enumerate}[\emph{(}i\emph{)}]
\item If $v=v_i$ and $v'=v_{(i+1)modt}$, then edge $e'$ is placed on
a new page $p'_j$, if and only if the forward edges incident to
block $B(v')$ are assigned to page $p_j$, $j=2,3$.
\item If $v=v_0$ and $v'=v_t$, then edge $e'$ is placed on a page
$p'_1$.
\end{enumerate}
\label{lem:lem_crossing_binding_forward}
\end{lemma}
\begin{proof}
\begin{inparaenum}[(i)]
\item We prove a stronger result. In particular, we prove that if
the forward edges incident to $B(v')$ are on $p_j$, then $e'$ can also
be placed on $p_j$ without crossings. Clearly, if this is true, then
the lemma follows, as one can always split one page into two. We
distinguish two cases based on whether $v=v_t$ or $v=v_i$ for $i<t$.

First, assume that $v=v_t$ and $v'=v_0$; refer to $e=(u_2,v_5)$ and
$e'=(u_3,v_2)$ in Fig.~\ref{fig:fig_crossing_forward_before}.
Let $B=B(v)$ and $B'=B(v')$. Then, $B=P$ and $B'$ is the
parent-block of $B$. W.l.o.g. assume that the boundary of $B'$ is on
$p_2$. We claim that $e'$ can be placed on $p_3$ (together with
forward edges of $B'$). To prove it, we show that $e'=(u',v_0)$ and
$(v_0,v_t)$ form a strong pair. First, observe that $(v_0,v_t)$ is
on $p_3$: $B'$ is on $p_2$, so $B$ is on $p_3$ and $(v_0,v_t)$ is an
edges of $B$. By
Lemma~\ref{lem:lem_crossing_ordering_1}.\ref{lem:lem_crossing_basics},
vertices $v_1,\dots,v_t$ and $u'$ appear in the same order from left
to right along the spine, so $(v_0,v_t)$ is nested by $e'$. So, by
Lemma~\ref{lem:lem_crossing_chord_separating} $e'$ and $(v_0,v_t)$
form a strong pair.

\begin{figure}[t]
    \centering
    \begin{minipage}[b]{.40\textwidth}
        \centering
        \subfloat[\label{fig:fig_crossing_forward_before}{}] 
        {\includegraphics[width=\textwidth,page=1]{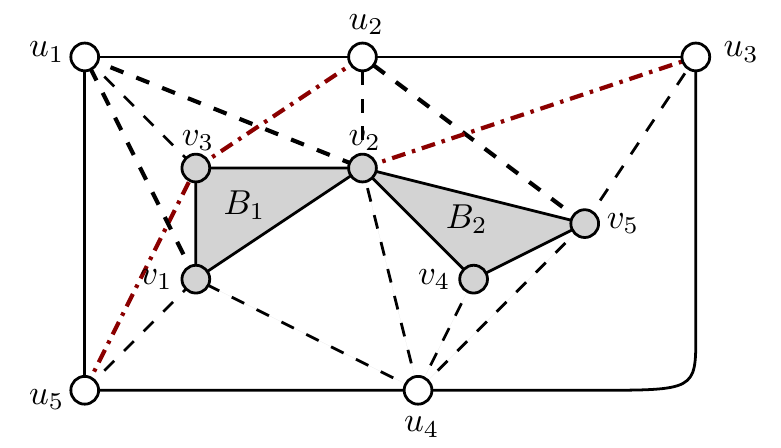}}
    \end{minipage}
    \hfill
    \begin{minipage}[b]{.40\textwidth}
        \centering
        \subfloat[\label{fig:fig_crossing_forward_after}{}] 
        {\includegraphics[width=\textwidth,page=2]{crossings}}
    \end{minipage}
    \caption{%
    (a)~Red edges indicate forward binding edges involved in crossings.
    (b)~Linear order and assignment  of edges to pages; The fat edge is assigned to $p_1'$; the dashed-dotted ones to $p_2'$ and $p_3'$.}
    \label{fig:fig_crossing_forward}
\end{figure}

In the case where $v=v_i$ and $v'=v_{i+1}$ for some $i<t$ (refer to
$e=(u_1,v_2)$ and $e'=(u_2,v_3)$ in
Fig.~\ref{fig:fig_crossing_forward_before}), we have that $B=B'=P$.
Suppose w.l.o.g. that $P$ is placed on $p_2$. We claim that
$e'=(u',v_{i+1})$ and $(u',v_i)$ form a strong pair. By
Lemma~\ref{lem:lem_crossing_quadrangle}, edge $(u',v_i)$ is a forward
edge of $G_P$ and is therefore placed on page $p_3$. Since vertices
$v_i$ and $v_{i+1}$ are consecutive along the spine, edges $e'$ and
$(u',v_i)$ form a strong pair and the lemma follows.

\item In this case (refer to $e=(u_1,v_1)$ and $e'=(u_5,v_3)$ in
Fig.~\ref{fig:fig_crossing_forward_before}), we have that $P$ is a
root-block and by Lemma~\ref{lem:lem_crossing_quadrangle} edge
$(u,u')\in E[G_P]$. Let $e'_1$ and $e'_2$ be two edges that are
assigned to the new page $p'_1$. We claim that they do not cross.
Assume that $e'_i=(u'_i,v'_i)$ crosses in $G$ with $e_i=(u_i,v_i)$
for $i=1,2$. Then, $(u_i,u'_i)\in E[G_P]$ is level edge of $L_0$,
for $i=1,2$. We assume w.l.o.g. that $u_2$ and $u'_2$ are not
between $u_1$ and $u'_1$ along the spine (if this was not true for
neither pair of vertices, then they would cross in $G$ and by
Lemma~\ref{lem:lem_crossing_inner_chord_auxiliary} one of them would
not belong to $G_P$). By
Lemma~\ref{lem:lem_crossing_chord_separating} edge $e'_1$ can't
cross with $e'_2$.
\end{inparaenum}
\end{proof}

Finally, we consider Case C.\ref{case:hop_binding} where $e'=(x,y)$
is a $2$-hop of level $L_1$ and $e=(u,z)$ is a binding edge of
$G_P$, where $x,y,z\in L_1$ and $u \in L_0$. Let $x$, $z$ and
$y$ belong to blocks $B_x$, $B_z$ and $B_y$, resp., that are not
necessarily distinct. By Lemma~\ref{lem:lem_crossing_quadrangle}, $x
\rightarrow z \rightarrow y$ is a path in $L_1$. So, $B_x$ and $B_y$
are at distance at most two on the block-tree of $G$. If $x$ and $y$
belong to the same block (that is, $B_x=B_z=B_y$), then $e$ is
called \emph{simple 2-hop}; see Fig.~\ref{fig:fig_crossing_simple}.
Suppose w.l.o.g. that $B_x$ precedes $B_y$ in the pre-order
traversal of the block-tree of $G$. Then, there exist two cases
depending on whether $B_x$ is an ancestor of $B_y$ on the
block-tree. If this is not the case, then $B_x$ and $B_y$ have the
same parent-block, say $B_p$. In this case, $e'$ is called
\emph{bridging $2$-hop}; see Fig.~\ref{fig:fig_crossing_bridging}.
Suppose now that $B_x$ is an ancestor of $B_y$. Then, the path $x
\rightarrow z \rightarrow y$ contains the leader of $B_y$, which is
either $x$ or $z$. By Lemma~\ref{lem:lem_crossing_quadrangle},
$(u,x)$ $(u,z)$ and $(u,y)$ exist in $G$. So, $u$ is either
$u_l(B_y)$ or $u_f(B_y)$. In the first subcase, $e'$ is called
\emph{forward $2$-hop}; see Fig.~\ref{fig:fig_crossing_forward}. In
the second subcase, since $B_x$ is ancestor of $B_y$ and the two
blocks are at distance at most two, if $B_x$ is the parent-block of
$B_y$, then $e'$ is called \emph{backward $2$-hop}; see
Fig.~\ref{fig:fig_crossing_backward}. Finally, if $B_x$ is the
grand-parent-block of $B_y$, then $e'$ is called \emph{long
$2$-hop}; see Fig.~\ref{fig:fig_crossing_long}.

\begin{figure}[t]
  \centering
  \begin{minipage}[b]{.30\textwidth}
    \centering
    \subfloat[\label{fig:fig_crossing_simple}{Simple $2$-hop}]
    {\includegraphics[width=\textwidth,page=1]{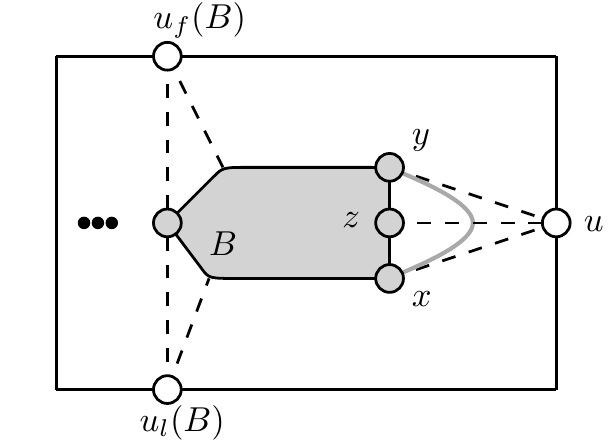}}
  \end{minipage}
  \begin{minipage}[b]{.30\textwidth}
    \centering
    \subfloat[\label{fig:fig_crossing_bridging}{Bridging $2$-hop}]
    {\includegraphics[width=\textwidth,page=2]{figures/2hops}}
  \end{minipage}
  \begin{minipage}[b]{.30\textwidth}
    \centering
    \subfloat[\label{fig:fig_crossing_forward}{Forward $2$-hop}]
    {\includegraphics[width=\textwidth,page=3]{figures/2hops}}
  \end{minipage}
  \begin{minipage}[b]{.30\textwidth}
    \centering
    \subfloat[\label{fig:fig_crossing_backward}{Backward $2$-hop}]
    {\includegraphics[width=\textwidth,page=4]{figures/2hops}}
  \end{minipage}
  \begin{minipage}[b]{.30\textwidth}
    \centering
    \subfloat[\label{fig:fig_crossing_long}{Long $2$-hop}]
    {\includegraphics[width=\textwidth,page=5]{figures/2hops}}
  \end{minipage}
  \caption{Different types of $2$-hops (drawn in gray).}
  \label{fig:fig_crossing_2hops}
\end{figure}

\begin{lemma}[Case C.\ref{case:hop_binding}]
All crossing $2$-hops can be assigned to seven pages in total.
\label{lem:two_hops}
\end{lemma}
\begin{proof}
In high level description, one can prove that all
simple $2$-hops can be embeded in any page that contains $2$-hops;
see Lemma~\ref{lem:lem_crossing_2hops_simpe}. All bridging $2$-hops
can be embedded in two new pages; see
Lemma~\ref{lem:lem_crossing_2hops_bridging}. Forward $2$-hops
can be embedded in one new page; see
Lemma~\ref{lem:lem_crossing_2hops_forward}. And, finally, backward
and long $2$-hops  can be embedded in two new pages each; see
Lemma~\ref{lem:lem_crossing_2hops_backward} and
\ref{lem:lem_crossing_2hops_long}, respectively. Summing up the
above, yields a total of seven pages for all $2$-hops.
\end{proof}

In the following, we show how to cope with each of the
aforementioned cases described above. In particular, let $e'=(x,y)$
be a $2$-hop of level $L_1$ that crosses a binding edge $e=(u,z)$ of
the underlying planar structure $G_P$. Then, $x,y,z\in L_1$ and $u
\in L_0$. Let $x$, $z$ and $y$ belong to blocks $B_x$, $B_z$ and
$B_y$, respectively. The different types of $2$-hops are summarized
as follows:

\begin{enumerate}
\item Vertices $x$ and $y$ belong to the same block: Simple $2$-hop
\item Vertices $x$ and $y$ belong to different blocks, say $B_x$ and $B_y$ respectively
\begin{enumerate}[2.1.]
\item Blocks $B_x$ and $B_y$ have the same parent-block $B_p$: Bridging $2$-hop
\item Block $B_x$ is an ancestor of $B_y$
\begin{enumerate}[2.2.1.]
\item Vertex $u$ is $u_l(B_y)$: Forward $2$-hop
\item Vertex $u$ is $u_f(B_y)$
\begin{enumerate}[2.2.2.1.]
\item $B_x$ is the parent-block of $B_y$: Backward $2$-hop
\item $B_y$ is the grand-parent-block of $B_y$: Long $2$-hop
\end{enumerate}
\end{enumerate}
\end{enumerate}
\end{enumerate}

Let $B$ be a block of level $L_1$ and assume that
$v_0,v_1,\dots,v_t$ are the vertices of $B$ as they appear in an
counterclockwise traversal of the boundary of $B$ starting from the
leader of $B$, say $v_0$; that is $v_0=\ell(B)$. In our proofs, we
use the notion of the \emph{trail} of inner vertex $v_i$, denoted by
$tr(v_i)$, which is recursively defined as follows. If $v_i$ is
identified by the first inner vertex, then $tr(v_i)=v_i$.
Otherwise, $tr(v_i)=tr(\ell(B)) \rightarrow v_1 \rightarrow \ldots
\rightarrow v_i$, $i=1,2,\ldots,t$. Intuitively, the trail of inner
vertex $v_i$ is the path that starts from the first inner vertex and
ends to $v_i$, which
\begin{inparaenum}[(i)] 
\item consists exclusively of level one vertices, and, 
\item traverses each intermediate block from the root block towards
$B$ always in counterclockwise direction. 
\end{inparaenum} 
In Fig.~\ref{fig:yannakakis-input}, the trail of vertex $v_9$ is
$tr(v_9) = v_1 \rightarrow v_2 \rightarrow v_3 \rightarrow v_6
\rightarrow v_8 \rightarrow v_9$ The \emph{trail of a block} is the
trail of its leader.
The following lemma follows from Yannakakis algorithm~\cite{Yan89}.

\begin{lemma}
Let $B$ be a block with leader $v_0$. Consider the trail $tr(B)$ of
$B$, vertices $u_f(B)$ and $u_l(B)$ as in
Fig.~\ref{fig:yannakakis_crossing_basic}. Then $G$ is partitioned
into regions, where the trail $tr(B)$ belongs to region $I$. Then,
along the spine vertices of region $I$ are to the left of vertices of $B$,
vertices of $B$ are to the left of vertices of region $II$ and
vertices of region $II$ are to the left of vertices of region $III$.
\label{lem:lem_crossing_basic_2} 
\end{lemma}

\begin{figure}[h]
    \centering
    \includegraphics[width=.4\textwidth,page=3]{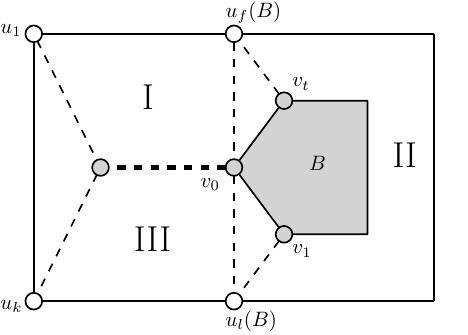}
    \caption{%
    Illustration of different regions I, II and III of
    Lemma~\ref{lem:lem_crossing_basic_2}; the trail of $B$ is
    drawn fat.}
    \label{fig:yannakakis_crossing_basic}
\end{figure}


\begin{lemma}
Let edges $e$ and $e'$ of $G$ cross such that $e$ is a binding edge
of $G_P$ and $e'$ a simple $2$-hop of $G$. Then, if $e'$ is placed
on the same page as any other $2$-hop $e''$, then $e'$ and
$e''$ do not cross.
\label{lem:lem_crossing_2hops_simpe}
\end{lemma}
\begin{proof}
Let $e'=(x,y)$ where $x,y\in L_1$ and $e=(u,z)$, where $u$ is a
vertex of $L_0$ and $z$ a vertex of $L_1$. By definition of simple
$2$-hops, vertices $x$, $z$ and $y$ belong to the same block $B$.
Then, for the leader $\ell(B)$ of $B$, it holds that
$\ell(B)\notin\left\{x,z,y\right\}$. By
Lemma~\ref{lem:lem_crossing_ordering_1}.\ref{lem:lem_crossing_basics},
vertices $x$, $z$ and $y$ are consecutive along the spine and appear
in this order from left to right. If $e'=(x,y)$ is placed on a page
$p_h$ and crosses with another $2$-hop $e''$ on $p_h$, then $e''$
has $z$ as one endpoint. Hence, $e'$ and $e''$ would cross in the
$1$-planar embedding of $G$; a contradiction.
\end{proof}


Recall that, by definition, if $e'=(x,y)$ is a bridging $2$-hop,
then $B_x$ and $B_y$ are at distance two, and they have the same
parent-block $B_P$ and a common leader, say $\ell$. Since $\ell$ is
a cut-vertex of $L_1$, which separates blocks $B_x$ and $B_y$, any
$x-y$ path on $L_1$ must go through $\ell$. Hence, for the path $x
\rightarrow z \rightarrow y$, we have that $z=\ell$. Also, from the
assumption that $B_x$ precedes $B_y$ along the spine and by
Lemma~\ref{lem:lem_crossing_quadrangle}, it follows that
$u=u_l(B_x)=u_f(B_y)$.

\begin{lemma}
Let $B$ be a block of $G$ with children-blocks $B_1$, $B_2$, \dots,
$B_s$, where for any $i<j$ all vertices of $B_i$ appear before all
vertices of $B_j$ along the spine. Then:
\begin{inparaenum}[\emph{(}i\emph{)}]
\item there is at most one bridging $2$-hop of $G$ with one endpoint on $B_1$,
\item there are at most two bridging $2$-hops of $G$ with one endpoint on $B_i$ for $i=2,\dots,s-1$ and
\item there is at most one bridging $2$-hop of $G$ with one endpoint on $B_s$.
\end{inparaenum} 
\label{lem:lem_crossing_2hops_bridging_auxiliary}
\end{lemma}
\begin{proof}
By definition, if $e'=(x,y)$ is bridging $2$-hop, then $B_x$ and
$B_y$ have the same parent-block $B_P$ and the same leader, say
$\ell$. Therefore, if $e'$ has one endpoint on one of the
children-blocks $B_i$ of $B$, then its other endpoint is also on
another child-block $B_j$ of $B$ ($i\neq j$). By
Lemma~\ref{lem:lem_crossing_ordering_1}, if $j>i+1$, i.e. $B_i$ and
$B_j$ are not consecutive child-blocks of $B$, then $e'$ would cross
with edges $(\ell,u_l(B_i))$, $(\ell,u_l(B_{i+1}))$, \dots,
$(\ell,u_l(B_{j-1}))$, i.e., with at least two edges of $G$,
contradicting $1$-planarity of $G$. So, $j=i-1$ or $j=i+1$ and
the lemma follows.
\end{proof}

\begin{lemma}
All bridging $2$-hops of $G$ can be placed on two new pages without
crossings.
\label{lem:lem_crossing_2hops_bridging}
\end{lemma}
\begin{proof}
Let $e_1'=(x_1,y_1)$ and $e_2'=(x_2,y_2)$ be two bridging $2$-hops
of $G$. W.l.o.g. assume that $x_1$, $x_2$, $y_1$, $y_2$ appear in
this order along the spine from left to right, i.e., edges $e_1'$
and $e_2'$ cannot be placed on the same page. Since $x_2$ appears
after vertex $x_1$, by Lemma~\ref{lem:lem_crossing_basic_2}, $x_2$
is either a vertex of $B_{x_1}$ with greater index than $x_1$, or
its block $B_{x_2}$ appears after block $B_{x_1}$.
Similarly, since $x_2$ appears before $y_1$ on the spine, by the
same lemma, $x_2$ is either a vertex of $B_{y_1}$ with
smaller index than $y_1$, or its block $B_{x_2}$ appears before
block $B_{y_1}$. By definition of bridging $2$-hops, 
$B_{x_1}$ and $B_{y_1}$ are consecutive children-blocks of a block
$B$ with the same leader. Hence, the only blocks that are between
$B_{x_1}$ and $B_{y_1}$ on the spine are descendant-blocks of
$B_{x_1}$ on the block-tree. Combining the above restrictions, we
distinguish three cases for $B_{x_2}$:
\begin{inparaenum}[(c.1)]
\item $B_{x_2}=B_{x_1}$ and $x_2$ appears after $x_1$,
\item $B_{x_2}$ is a descendant of $B_{x_1}$, and
\item $B_{x_2}=B_{y_1}$ and $x_2$ appears before $y_1$
\end{inparaenum}

In the first and third case, edges $e_1'$ and $e_2'$ are bridging
$2$-hops with one endpoint on the same block. In the second case,
again by definition, $B_{y_2}$ is also a descendant-block of
$B_{x_1}$ (since $B_{x_2}$ and $B_{y_2}$ have the same
parent-block). Then, $B_{y_2}$ appears before  $B_{y_1}$ on the
ordering of blocks; a contradiction. Hence, if two bridging $2$-hops
$e_1'$ and $e_2'$ cannot be placed on the same page, then they have
an endpoint on the same block.

Consider now an auxiliary graph where vertices are bridging $2$-hops
of $G$ and an edge exists if the bridging $2$-hops corresponding to
its endpoints cannot be placed on the same page, i.e., if the two
bridging $2$-hops have an endpoint on the same block in $G$.
By Lemma~\ref{lem:lem_crossing_2hops_bridging_auxiliary}, it follows
that the auxiliary graph consists of disjoint paths and is therefore
bipartite. We assign bridging $2$-hops that correspond to vertices
of the first (second, resp.) bipartition in the auxiliary graph to
first (second, resp.) page of the two available ones. It is clear
that bridging $2$-hops on the same page do not cross, concluding the
proof of the lemma.
\end{proof}


Recall that, by definition, if $e'=(x,y)$ is a forward $2$-hop, then
$B_x$ is an ancestor of $B_y$ and vertex $u$ is $u_l(B_y)$. Also,
since $x \rightarrow z \rightarrow y$ is a path of $G$ and $x$, $y$
belong to different blocks, it follows that the index of $y$ on
$B_{y}$ is either $1$ or $2$.

\begin{lemma}
Let $v_1$ and $v_2$ be two vertices of $L_1$ such that $v_1$ appears
before $v_2$ on the spine. Let $v$ be the last common vertex of 
trails $tr(v_1)$ and $tr(v_2)$. Then, all vertices of $tr(v_1)$
after $v$ precede vertices of $tr(v_2)$ after $v$.
\label{lem:lem_crossing_2hops_forward_auxiliary_1} 
\end{lemma}
\begin{proof}
Let $B_1=B(v_1)$, $B_2=B(v_2)$ and $B=B(v)$. Since $v$ is the last
common vertex of $tr(v_1)$ and $tr(v_2)$, it follows that $B$ is the
last common ancestor of $B_1$ and $B_2$ on the block-tree. Then, on
the ordering of blocks, all blocks on the path defined from $B$ to
$B_1$ appear before all blocks on the path from $B$ to $B_2$. Hence,
$tr(v_1)$ and $tr(v_2)$ are identical up to vertex $v$ and vertices
of $tr(v_1)$ after $v$ precede vertices of $tr(v_2)$ after $v$. 
\end{proof}

\begin{lemma}
Let $e'=(x,y)$ be a forward $2$-hop. Then, $tr(y) = tr(x)
\rightarrow z \rightarrow y$.
\label{lem:lem_crossing_2hops_forward_auxiliary_2} 
\end{lemma}
\begin{proof}
It suffices to prove that $x,z,y$ appear in this order from left to
right along the spine. Since $y$ is either the first or second
vertex of $B_y$, either $z=\ell(B_y)$ or $z$ is the first vertex of
$B_y$ respectively. In both cases, the desired property holds.
\end{proof}

\begin{lemma}
All forward $2$-hops of $G$ can be placed on a new page without crossings.
\label{lem:lem_crossing_2hops_forward}
\end{lemma}
\begin{proof}
Let $e_1'=(x_1,y_1)$ and $e_2'=(x_2,y_2)$ be two forward $2$-hops of
$G$. W.l.o.g. assume that $x_1$, $x_2$, $y_1$, $y_2$ appear in this
order along the spine from left to right, i.e., edges $e_1'$ and
$e_2'$ cannot be placed on the same page. Consider the trails of
$y_1$ and $y_2$. By
Lemma~\ref{lem:lem_crossing_2hops_forward_auxiliary_2},
$tr(y_1)=tr(x_1) \rightarrow z_1 \rightarrow y_1$ and
$tr(y_2)=tr(x_2) \rightarrow z_2 \rightarrow y_2$. Let $v_x$ be the
last common vertex of the trails $tr(x_1)$ and $tr(x_2)$ and $v_y$
the last common vertex of the trails $tr(y_1)$ and $tr(y_2)$.
By Lemma~\ref{lem:lem_crossing_2hops_forward_auxiliary_2}, if $v_x$
is not $x_1$, then all three vertices $x_1$, $z_1$ and $y_1$ will
appear before $x_2$, $z_2$ and $y_2$ (as in this case we would have
$v_x=v_y$); a contradiction. Hence, $v_x=x_1$ and $v_y\in
\left\{x_1,z_1,y_1\right\}$. Similarly, let $v$ be the last common
vertex of the trails $tr(x_2)$ and $tr(y_1)$. By
Lemma~\ref{lem:lem_crossing_2hops_forward_auxiliary_2}, if $v$ is
not $x_2$, then all three vertices $x_2$, $z_2$ and $y_2$ will
appear before $y_1$ (as in this case we would have $v=v_y$);
contradiction. Hence $v=x_2$ and $v_y\in \left\{x_2,z_2,
y_2\right\}$. From the above, it follows that $v_y\in
\left\{x_1,z_1,y_1\right\}$. By the order of the vertices along the
spine, it follows that $v_y\notin \{x_1,y_2\}$. If $v_y=z_1$, then
$v_y=x_2$ also (otherwise, if $v_y=z_2$ vertex $x_2$ would also be
on the common part of the trails, and inevitably it would be
$x_2=x_1$;contradiction). On the other hand, if $v_y=y_1$, then
$v_y=z_2$ and $z_1=x_2$ (since they are the unique vertex on the
trail before $v_y$). In the first case, $e'_1$ and $e'_2$ would also
cross in the $1$-planar embedding of $G$. The second case
contradicts Lemma~\ref{lem:lem_crossing_quadrangle}, since edges
$(z_1,u_2)$, $(z_1,x_1)$ and $(z_1,y_1)$  are all incident to
$z_1$ (recall $z_1=x_2$).
\end{proof}


By definition, if $e'=(x,y)$ is a backward $2$-hop, then vertex
$u=u_f(B_y)$ and $B_x$ is the parent-block of $B_y$ and $y$ is
either the last or the second-to-last vertex of $B_y$.

\begin{lemma}
If $e'=(x_1,y_1)$ and $e''=(x_2,y_2)$ are two backward $2$-hops,
then $B(y_1)\neq B(y_2)$.
\label{lem:lem_crossing_2hops_backward_auxiliary}
\end{lemma}
\begin{proof}
For a backward $2$-hop $e'=(x,y)$ we say that the last edge of $B_y$
is \emph{covered} by edge $e'$. By
Lemma~\ref{lem:lem_crossing_quadrangle}, the last edge of $B_y$ can
be covered by at most one backward $2$-hop. So, $B(y_1)\neq B(y_2)$
holds for $e'_1$ and $e'_2$.
\end{proof}

\begin{lemma}
All backward $2$-hops of $G$ can be placed on two new pages without
crossings.
\label{lem:lem_crossing_2hops_backward}
\end{lemma}
\begin{proof}
Let $e_1'=(x_1,y_1)$ and $e_2'=(x_2,y_2)$ be two backward $2$-hops
of $G$. W.l.o.g. assume that $x_1$, $x_2$, $y_1$, $y_2$ appear in
this order along the spine from left to right, i.e., edges $e_1'$
and $e_2'$ cannot be placed on the same page. Since $x_2$ appears
after vertex $x_1$, by Lemma~\ref{lem:lem_crossing_basic_2}, 
$x_2$ is either a vertex of $B_{x_1}$ with greater index than
$x_1$, or its block $B_{x_2}$ appears after block $B_{x_1}$.
Similarly, since $x_2$ appears before $y_1$ on the spine, by the
same lemma, $x_2$ is either a vertex of $B_{y_1}$ with smaller index
than $y_1$, or its block $B_{x_2}$ appears before block $B_{y_1}$.
Combining the above restrictions, we distinguish three cases for
$B_{x_2}$:
\begin{inparaenum}[(c.1)]
\item $B_{x_2}=B_{x_1}$ and $x_2$ appears after $x_1$,
\item $B_{x_2}$ is a descendant of $B_{x_1}$, and
\item $B_{x_2}=B_{y_1}$ and $x_2$ appears before $y_1$.
\end{inparaenum}

In the first two cases, since $B_{y_2}$ is child-block of $B_{x_2}$,
$B_{y_2}$ will appear before $B_{y_1}$ and $y_2$ will be to the left
of $y_1$; a contradiction. Hence, if two backward $2$-hops $e_1'$
and $e_2'$ cannot be placed on the same page, then $B_{x_2}=B_{y_1}$.

We say that a backward $2$-hop $e'=(x,y)$ is \emph{charged} to block
$B_x$ (i.e., to the parent-block). Then, whenever $e_1'$ and $e_2'$
cannot be placed on the same page, we have that they are charged to
blocks of consecutive levels on the block-tree. Then, we can place
backward $2$-hops that are charged on blocks of odd level on one of
the two new available pages and those charged on blocks of even
level on the other available page. It is clear that backward
$2$-hops on the same page do not cross, concluding the proof of the
lemma.
\end{proof}


By definition, if $e'=(x,y)$ is a long $2$-hop, then 
$u=u_f(B_y)=u_f(B_z)$, $B_y$ is the first child-block of $B_z$ and
$B_z$ is a child-block of $B_x$.

\begin{lemma}
Let $e'=(x,y)$ be a long $2$-hop of $G$. We say that $B_z$ is the
\emph{middle} block of $e'$, and $B_y$ the \emph{ending} block of
$e'$. Then, a block $B$ can be middle (ending) block of at most one
long $2$-hop $e'$. Also, $B$ cannot be middle block of a long
$2$-hop $e'$ and ending block of another long $2$-hop $e''$ at the same
time.
\label{lem:lem_crossing_2hops_long_auxiliary}
\end{lemma}
\begin{proof}
The proof is similar to the one of
Lemma~\ref{lem:lem_crossing_2hops_backward_auxiliary}. We say that
the last edge of a block $B$ is \emph{covered} by a long $2$-hop
$e'=(x,y)$, if $B$ is the middle or ending block of $e'$. By
Lemma~\ref{lem:lem_crossing_quadrangle} the last edge of $B$ can be
covered by at most one long $2$-hop, and the lemma follows.
\end{proof}

\begin{lemma}\label{lem:lem_crossing_2hops_long}
All long $2$-hops of $G$ can be placed on two new pages without
crossings.
\end{lemma}
\begin{proof}
Let $e_1'=(x_1,y_1)$ and $e_2'=(x_2,y_2)$ be two long $2$-hops of
$G$. W.l.o.g. assume that $x_1$, $x_2$, $y_1$, $y_2$ appear in this
order along the spine from left to right, i.e., edges $e_1'$ and
$e_2'$ cannot be placed on the same page. Since $x_2$ appears after
vertex $x_1$, by Lemma~\ref{lem:lem_crossing_basic_2}, $B_y$ is the
first child-block of $B_z$, we have that $x_2$ is either a vertex of
$B_{x_1}$ with greater index than $x_1$, or its block $B_{x_2}$
appears after block $B_{x_1}$. Similarly, since $x_2$ appears before
$y_1$ on the spine, by the same lemma, we have that $x_2$ is either
a vertex of $B_{y_1}$ with smaller index than $y_1$, or its block
$B_{x_2}$ appears before block $B_{y_1}$. Combining the above
restrictions, we distinguish four cases for $B_{x_2}$:
\begin{inparaenum}[(c.1)]
\item $B_{x_2}=B_{x_1}$ and $x_2$ appears after $x_1$,
\item $B_{x_2}$ is a descendant of $B_{x_1}$ before $B_{z_1}$,
\item $B_{x_2}=B_{z_1}$, and
\item $B_{x_2}=B_{y_1}$ and $x_2$ appears before $y_1$.
\end{inparaenum}

In the first two cases, since $B_{y_2}$ is grand-child-block of
$B_{x_2}$, $B_{y_2}$ will appear before $B_{y_1}$ and $y_2$ will be
to the left of $y_1$; a contradiction. Hence, if two long $2$-hops
$e_1'$ and $e_2'$ cannot be placed on the same page, then
$B_{x_2}=B_{y_1}$ or $B_{x_2}=B_{z_1}$. We say that a long $2$-hop
$e_1'$ is \emph{indirectly in conflict with} $e_2'$ if
$B_{x_2}=B_{y_1}$ and \emph{directly in conflict with} $e_2'$ if
$B_{x_2}=B_{z_1}$. Note that
Lemma~\ref{lem:lem_crossing_2hops_long_auxiliary} in terms of
conflicts means that a long $2$-hop cannot be simultaneously
indirectly and directly in conflict with other $2$-hops. In the
following, we use induction on the number of blocks to show
that two pages are sufficient.

Start with block $B_1$, the root-block of the block-tree. Long
$2$-hop edges $e'=(x,y)$ with $B_x=B_1$ are not in conflict with
each other and can, therefore, be placed on the same page. Suppose
that we have placed all long $2$-hop edges $e'=(x,y)$ with $B_x\leq
B_s$ on two pages, say $p^1$ and $p^2$.

Let $e'=(x,y)$ with $B_x=B_{s+1}$ be a long $2$-hop. We claim that
it can be placed on one of the two pages without introducing any
crossings. We have that $e'$ is indirectly in conflict with a long
$2$-hop $e_1'=(x_1,y_1)$, if $B_x=B_{y_1}$ and directly in conflict
with a long $2$-hop $e_2'=(x_2,y_2)$, if $B_x=B_{z_2}$. By
Lemma~\ref{lem:lem_crossing_2hops_long_auxiliary}, $e'$ cannot be
simultaneously indirectly and directly in conflict with other long
$2$-hops. So, assume first that $e'$ is only indirectly in conflict
with other long $2$-hops of $G$. As already stated, if $e'$ is
indirectly in conflict with $e_1'=(x_1,y_1)$, then $B_x=B_{y_1}$.
Then, $B_{z_1}$ is the parent-block of $B_x$ and $B_{x_1}$ the
grand-parent-block of $B_x$. These blocks are uniquely defined and
$e_1'$ is the only long $2$-hop that $e'$ is indirectly in conflict
with. Clearly, if $e_1'$ is on page $p^i$, then $e'$ can go on the
opposite page $p^j$, where $i\neq j$. In the case where  $e'$ is
only directly in conflict with other long $2$-hops of $G$, we have
that if $e_2'=(x_2,y_2)$ is such an edge, then $B_x=B_{z_2}$.
By Lemma~\ref{lem:lem_crossing_2hops_long_auxiliary}, $B_x$ can be
middle block of at most one long $2$-hop, that is, $e_2'$ is the
only long $2$-hop that $e'$ is directly in conflict with. So, if
$e_2'$ is on page $p^i$, then $e'$ can go on the opposite page
$p^j$, where $i\neq j$. By induction, it follows that all long
$2$-hop edges of $G$ can be placed on two new pages.
\end{proof}

\subsection{The Multi-Level Case}
\label{sec:multi-layered}
In this section, we consider the general case, according to which
the given $1$-planar graph $G$ consists of more than two
levels, say $L_0,L_1,\ldots,L_\lambda$; $\lambda \geq 2$. We still
assume that $\Gamma(G)$ is simple internally maximal $1$-planar
drawing and has no crossings on its unbounded face.

\begin{lemma}\label{thm:simplemaximal}
Any simple internally maximal $1$-planar graph $G$ with $\lambda \geq
2$ levels and no crossings incident to its unbounded face admits a
book embedding on $34$ pages.
\end{lemma}
\begin{proof}
We first embed in $5$ pages the underlying planar structure $G_P$ of
$G$ using the algorithm of Yannakakis~\cite{Yan89}. This
implies that all vertices of a block of level $i$, except possibly
for its leader, are between two consecutive vertices of level $i-1$,
$i=1, \ldots, \lambda$.  So, for outer crossing chords that are
involved in crossing with level edges of $G_P$
(Case~C.\ref{case:chord_chord}), one universal page (denoted by
$up_c$ in Lemma~\ref{lem:lem_crossing_inner_chord}.i) suffices,
since such chords are not incident to block-leaders.

Next, we consider the outer crossing chords that are involved in
crossings with binding edges or bridge-blocks of $G$ (Cases
C.\ref{case:chord_binding} and C.\ref{case:chord_bridge}).
Recall that such a chord $c_{i,j}=(v_i,v_j)$ of a block $B$ is on
the same page as the boundary and the non-crossing chords of $B$.
The path $P[v_i\rightarrow v_j]$ on the boundary of $B$ joins the
endpoints of the crossing chord. Hence, if another edge of the same
page crosses with $c_{i,j}$, it must also cross with an edge of $B$;
a contradiction. Therefore, such chords do not require additional
pages.

For binding edges of Case C.\ref{case:forward_forward}, $5$ pages in
total suffice; one page for each page of $G_P$. For binding edges of
Case C.\ref{case:forward_back}, however, we need a different
argument: Since a binding edge between levels $L_{i+2}$ and
$L_{i+1}$ can not cross with a binding edge between levels $L_{i-1}$
and $L_{i-2}$, $i = 2,\ldots,\lambda-2$, it follows that binding
edges that bridge pairs of levels at distance at least $3$ are
independent. So, for binding edges of Case
C.\ref{case:forward_back} we need a total of $3*3=9$ pages.

Similarly, all blocks of level $i+1$ that are in the interior of a
certain block of level $i$ are always between two consecutive
vertices of level $i-1$, $i=1, 2, \ldots, \lambda-1$. Hence,
$2$-hops that are by at least two levels apart in the peeling order
are independent, which implies that for $2$-hops we need a total of
$2*7=14$ pages (Case~C.\ref{case:hop_binding}). Summing up we need
$5+1+5+9+14=34$ pages for a multi-level simple maximal $1$-planar
graph $G$.
\end{proof}

\subsection{Coping with non-maximal $1$-planar graphs:}
\label{subsec:nonmaximal}
In case of a planar topological graph, one can add edges to make it
maximal planar (and simultaneously preserve simplicity).
Unfortunately, this is not always possible in the case of $1$-planar
topological graphs (without loosing simplicity), as the produced
graph may contain multiedges. We can assume, however, that all
multi-edges are crossing-free, that is, they belong to the
underlying planar structure. Indeed, if a multi-edge contains an
edge that is involved in a crossing, then this particular edge can
be safely removed from the graph (as it can be ``replaced'' by any
of the corresponding crossing-free edges that were added during the
triangulation at the beginning of the algorithm). In the following,
we describe how to cope with a non-simple maximal $1$-planar graph
$G$.

Let $(v,w)$ be a double edge of $G$. Denote by $G_{in}[(v,w)]$ the
so-called \emph{interior subgraph} of $G$ w.r.t.
$(v,w)$ bounded by the double edge $(v,w)$ in $\Gamma(G)$. By
$G_{ext}[(v,w)]$ we denote the so-called \emph{exterior subgraph} of
$G$ w.r.t. $(v,w)$, derived from $G$ by substituting $G_{in}[(v,w)]$
by a single edge; see Fig.~\ref{fig:multiedges}. Clearly,
$G_{ext}[(v,w)]$ stays maximal $1$-planar and simultaneously has
fewer multiedges than $G$. So, it can be recursively embedded. The
base of the recursion corresponds to a simple maximal $1$-planar
graph, that is embedded according to Lemma~\ref{thm:simplemaximal}.

On the other hand, however, we can not assure that the interior
subgraph has fewer multiedges than $G$. Our aim is to modify it
appropriately, so as to reduce the number of its multiedges by one.
To do so, we will ``remove'' the multiedge $(v,w)$ that defines the
boundary of $G_{in}[(v,w)]$, so as to be able to recursively embed
it (again we seek to employ Lemma~\ref{thm:simplemaximal} in the
base of the recursion). Let $e_i(v)$ ($e_i(w)$, resp.) be the $i$-th
edge incident to vertex $v$ ($w$, resp.) in clockwise direction that
is strictly between the two edges that form the double edge $(v,w)$.
We replace vertex $v$ ($w$, resp.) by a path of $d(v)$ ($d(w)$,
resp.) vertices, say $v_1,v_2,\dots,v_{d(v)}$
($w_1,w_2,\dots,w_{d(w)}$, resp.), such that vertex $v_i$ ($w_i$,
resp.) is the endpoint of edge $e_i(v)$ ($e_i(w)$, resp.). Let
$\overline{G}_{in}[(v,w)]$ be the implied graph, which can be
augmented to internally maximal $1$-planar and simultaneously has
fewer multiedges than $G_{in}[(v,w)]$. Since
$\overline{G}_{in}[(v,w)]$ has no crossings incident to its
unbounded face, it can be embedded recursively.

\begin{figure}[t]
   \centering
   \begin{minipage}[b]{.32\textwidth}
     \centering
     \subfloat[\label{fig:multiedges_1}{$G_{in}[(v,w)]$}] 
     {\includegraphics[width=\textwidth,page=1]{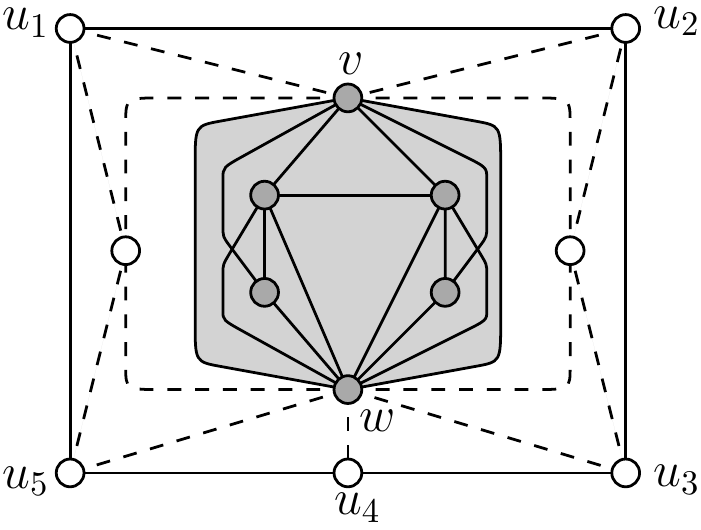}}
   \end{minipage}
   \hfill
   \begin{minipage}[b]{.32\textwidth}
     \centering
     \subfloat[\label{fig:multiedges_2}{$G_{ext}[(v,w)]$}]
     {\includegraphics[width=\textwidth,page=2]{figures/multiedges}}
   \end{minipage}
   \hfill
   \begin{minipage}[b]{.32\textwidth}
     \centering
     \subfloat[\label{fig:multiedges_2}{$\overline{G}_{in}[(v,w)]$}] 
     {\includegraphics[width=\textwidth,page=3]{figures/multiedges}}
   \end{minipage}
   \caption{Illustration of the decomposition in case of multiedges.}
   \label{fig:multiedges}
\end{figure}

It remains to describe how to plug the embedding of
$\overline{G}_{in}[(v,w)]$ to the embedding of $G_{ext}[(v,w)]$.
Suppose that $(v,w)$ of $G_{ext}[(v,w)]$ is on page $p$. Clearly,
$p$ is one of the pages used to embed the planar structure of
$G_{ext}[(v,w)]$, since $(v,w)$ is not involved in crossings in
$G_{ext}[(v,w)]$. We assume w.l.o.g. that the boundary of
$\overline{G}_{in}[(v,w)]$ is also on page $p$. Since $(v,w)$ is
already present in the embedding of $G_{ext}[(v,w)]$, it suffices to
plug in the embedding of $G_{ext}[(v,w)]$ only the interior of
$\overline{G}_{in}[(v,w)]$, which is the same as the one of
$G_{in}[(v,w)]$.

Suppose w.l.o.g. that in the embedding of $G_{ext}[(v,w)]$ vertex
$v$ appears before $w$ along the spine from left to right. Then, we
place the interior subgraph of $\overline{G}_{in}[(v,w)]$ to the
right of $v$. The edges connecting the interior of
$\overline{G}_{in}[(v,w)]$ with $v$ are assigned to page $p$, while
the ones connecting it to $w$ on page $p'$, which is a new page. In
such a way, we need $5$ more pages, one for each page of the planar
structure.

Next, we prove that no crossings are introduced. By restricting the
boundary of $\overline{G}_{in}[(v,w)]$ on page $p$, all edges
incident to $v$ towards $\overline{G}_{in}[(v,w)]$ become back
edges of $\overline{G}_{in}[(v,w)]$. So, edges that join $v$
with vertices in the interior of $\overline{G}_{in}[(v,w)]$ do not cross
with other edges in the interior of $\overline{G}_{in}[(v,w)]$.
Since $\overline{G}_{in}[(v,w)]$ is placed next to $v$, edges
incident to $v$ do not cross with edges of $G_{ext}[(v,w)]$ on page
$p$. Similarly, we can prove that edges incident to $w$ towards the
interior of $\overline{G}_{in}[(v,w)]$ do not cross with other edges
in the interior of $\overline{G}_{in}[(v,w)]$ on page $p'$.
We claim that potential crossings posed by different double edges,
say $(v,w)$ and $(v',w')$ with $v$ to the left of $v'$, do not
occur. In the case where $v \neq v'$ such a crossing would imply
that $(v,w)$ and $(v',w')$ cross; a contradiction.
If on the other hand $v=v'$ (and w.l.o.g. $w$ to the left of $w'$),
then it suffices to place $\overline{G}_{in}[(v,w)]$ before
$\overline{G}_{in}[(v',w')]$.

\subsection{Coping with crossings on graph's unbounded face:}
If there exist crossings incident to the unbounded face of $G$,
then, when we augment $G$ to maximal $1$-planar, we must also
triangulate the unbounded face of the planarized graph implied by
replacing all crossings of $G$ with crossing vertices (recall the
first step of our algorithm). This procedure may lead to a maximal
$1$-planar graph whose unbounded face is a double edge, say $(v,w)$.
In this case, $G_{ext}[(v,w)]$ consists of two vertices and a single
edge between them; $\overline{G}_{in}[(v,w)]$ is treated as
described above. So, we are now ready to state the main result of
our work.

\begin{theorem}
Any $1$-planar graph admits a book embedding in a book of $39$
pages.
\end{theorem}

\section{Conclusions and Open Problems} 
In this paper, we proved that $1$-planar graphs can be embedded in
books with constant number of pages, improving the previous known
bound, which was $O(\sqrt{n})$ for graphs with $n$ vertices. To keep
the description simple, we decided not to ``slightly'' reduce the
page number by more complicated arguments. So, a reasonable question
is whether the number of pages can be further reduced, e.g., to less
than $20$ pages. This is question of importance even for
\emph{optimal $1$-planar graphs}, i.e., graphs with $n$ vertices and
exactly $4n-8$ edges. Other classes of non-planar graphs that fit in
books with constant number of pages are also of interest.

\paragraph{Acknowledgement:} We thank S.~Kobourov and
J.~Toenniskoetter for useful discussions. We also thank David R.
Wood for pointing out an error in the first version of this paper.

\bibliographystyle{plain}
\bibliography{book}
\newpage
\appendix

\end{document}